\documentclass[11pt]{amsart}
\usepackage{graphicx}
\usepackage{amsmath,amsthm,bbm,times,mathtools,cite}
\usepackage{amssymb}
\usepackage{layout}
\usepackage{bm}
\usepackage{bbold,enumitem}

\usepackage{pst-plot,ulem}
\usepackage{hyperref}

\newcommand{\ran}{\operatorname{ran}}

\newcommand{\gap}{\operatorname{Gap}}

\newcommand{\spa}{\operatorname{span}}
\def\tr{\mathop{\mathrm{tr}}\nolimits}

\newtheorem{theorem}{Theorem}
\newtheorem{corollary}[theorem]{Corollary}
\newtheorem{lemma}[theorem]{Lemma}

\numberwithin{equation}{section}
\numberwithin{theorem}{section}

\newtheorem*{key}{Key Lemma}
\newtheorem{assum}[theorem]{Assumption}

\date{July 11, 2025}

\begin{document}
	
	\title[Spectral comparison]{Recursive spectral relations and the charge versus neutral gap in fractional quantum Hall systems}

	\author[M.~Lemm]{Marius Lemm}
	\address{\textnormal{(Marius Lemm)} Department of Mathematics, Eberhard Karls University of T\"ubingen, 72076 T\"ubingen, Germany.}
\
	
	\author[B.~Nachtergaele]{Bruno Nachtergaele} 
	\address{\textnormal{(Bruno Nachtergaele)} Department of Mathematics	and Center for Quantum Mathematics and Physics, University of California,  Davis, CA  95616-8633, USA}

	\author[S.~Warzel]{Simone Warzel}
	\address{\textnormal{(Simone Warzel)}  Department of Mathematics and Physics, TU M\"{u}nchen, 85747 Garching, Germany and Munich Center for Quantum Science and Technology, 80799 M\"unchen, Germany}

	\author[A.~Young]{Amanda Young}
	\address{\textnormal{(Amanda Young)} Department of Mathematics, University of Illinois Urbana-Champaign, Urbana, IL 61801, USA}
	
	\begin{abstract}
We consider quantum lattice Hamiltonians and derive recursive spectral relations bridging successive particle number sectors. One relation gives conditions under which the charge gap dominates the neutral gap. We verify these conditions under a triad of symmetries (translation-invariance, charge and dipole conservation) that are present, e.g., in periodic fractional quantum Hall systems. Thus, this gap domination, previously observed numerically, is a universal feature imposed by symmetry.  A second relation yields a new induction-on-particle-number method for deriving spectral gaps. The results cover both bosons and fermions.

	\end{abstract}
	
	\maketitle

\section{Introduction}
	
Controlling the low-lying energy spectrum of quantum many-body systems is a central problem of condensed-matter physics with few mathematical tools available. In this note, we consider a broad class of bosonic or fermionic lattice Hamiltonians, and we present two recursive spectral relations that link the two lowest energy eigenvalues in successive particle number sectors. 	

\begin{itemize}
\item[(I)]
Recursive spectral relation (I) compares the minimal energy cost of creating a particle-hole pair (``charge gap'') with that of adding a particle-number conserving excitation (``neutral gap''). The relation becomes exploitable when combined with symmetries that are present, e.g., in fractional quantum Hall systems, as we explain below.

\item[(II)] Recursive spectral relation  (II) compares the excited eigenvalues in successive particle number sectors. By iteration, we obtain a novel inductive method for deriving spectral gaps in many-body systems, a notorious problem with few available tools. 
	 \end{itemize}

Both relations (I) and (II) hold for very general Hamiltonians. However, they depend on norms of system-dependent overlap matrices which need to be controlled through system specific properties. Roughly speaking, our relations can be interpreted  as reducing spectral comparisons to calculating ``angles'' between many-body subspaces.\\

Our main motivation for the recursive spectral relations are the class of pseudopotential Hamiltonians that model Fractional Quantum Hall Systems (FQHS). 
Physically, these systems are characterized by the appearance of a maximal fractional filling at which the system becomes incompressible. Their so-called ``charge gap'' is the associated energy jump in the ground-state energy when the filling is increased beyond the maximal one. By contrast, the ``neutral gap'' is the energy difference between the first excited and the ground-state energy at maximal filling. In FQHS, the charge gap has been numerically observed to be larger than the neutral gap, whose eigenstate is identified with the magneto-roton~\cite{PhysRevLett.54.237, Girvin:1986aa,morf2002excitation,bonderson2011numerical, PhysRevLett.112.046602, Lu:2024bh}. This inequality among the two gaps is shared by the nematic fractional Hall phases in which the charge gap stays open in the thermodynamic limit while the neutral gap closes~\cite{Fradkin:1999aa, Regnault:2017aa, PhysRevResearch.2.033362, Pu:2024aa}. By leveraging Relation (I) in combination with the fundamental symmetries of such systems (translation-invariance, conservation of particle number, and of dipole moment), we are able to show that the observed incompressibility and the domination of the charge gap over the neutral gap are in fact universal features of FQHS systems.

Relation (II) connects to another famous open problem in FQHS, known as the ``Haldane's FQHS conjecture''. The question is to prove the positivity of the excitation gap above the ground state uniformly in the system size. A new approach is an induction by particle number exploiting Relation (II). It adds to the very limited toolkit for deriving spectral gaps which  consists of methods based on calculating norms involving ground state projections \cite{fannes1992finitely,nachtergaele1996spectral,nachtergaele:2016b}  and finite-size criteria \cite{knabe1988energy,gosset2016local,lemm2019spectral}.
While combining techniques from \cite{knabe1988energy,nachtergaele:2016b} led to a proof of a uniform spectral gap for truncated pseudopotential Hamiltonians  \cite{NWY2020b,nachtergaele:2021a,warzel:2021,warzel:2022}, establishing the full conjecture even with the new methods presented here would require considerable further work. 
%

%
%
%

		
			\subsection{Organization of the paper}
	
	This note is organized as follows. We first lay out the general many-body setup, in which we establish our bounds on the lowest eigenvalues and present the gap comparison method. We then explain in Section~\ref{sec:cngap} how one of these bounds can be used to establish a relation between the charge and neutral gaps. Section~\ref{sec:FQHS} is devoted to FQHS and contains our main result comparing the charge and neutral gaps in this setup.

\section{Recursive spectral relations} 
In this section, we work in a general setup of quantum lattice Hamiltonians and establish the recursive spectral Relations (I) and (II).

\subsection{Setup}
We consider spinless Fermi or Bose systems with a finite-dimensional one-particle space. Without loss of generality, one may take the latter as the square summable sequences $ \ell^2(\Lambda) $ over a finite set of sites, $\Lambda$.   The canonical creation $  (a_x^\dagger) $ and annihilation  $ (a_x )$ operators for each site  $ x \in \Lambda $ are defined on the corresponding fermionic (+) or bosonic (-) Fock space $ \mathcal{F}^\pm $. They obey the canonical anticommutation or commutation rules:
	\begin{equation}\label{eq:canonical}
		[a_x,a_y^\dagger]_\pm = a_x a_y^\dagger \pm  a_y^\dagger a_x = \delta_{x,y} , \qquad [a_x,a_y]_\pm = [a_x^\dagger,a_y^\dagger]_\pm = 0 . 
	\end{equation}
	We study Hamiltonians that are sums of self-adjoint  $ m $-body operators,
	\begin{align}\label{eq:strucH}
		& H = \sum_{m=m_0}^{m_1} H^{(m)}  \quad \mbox{with}\\
		& H^{(m)}= \mkern-10mu \sum_{\substack{  x_1\dots x_m\in\Lambda \\  y_1 \dots y_m\in\Lambda} } W_{x_1\dots x_m}^{y_1 \dots y_m}\ 
		a^\dagger_{x_1} \dots a^\dagger_{x_m}  a_{y_m} \dots a_{y_1}  \notag
	\end{align}
	where  $ 0 \leq m_0 \leq m_1 $, and the $ m $-body coefficients $  W_{x_1\dots x_m}^{y_1 \dots y_m} \in \mathbb{C} $ are constrained by the requirement that all $ H^{(m)} $ are  self-adjoint and bounded from below.  
	
	The Fock space decomposes into a direct sum of (finite-dimensional) $ n$-particle Hilbert spaces,
	$$
	\mathcal{F}^\pm = \bigoplus_{n\geq 0 } \mathcal{F}_n^\pm , \qquad N := \sum_{x \in \Lambda} a_x^\dagger a_x = \bigoplus_{n\geq 0 }  n \mathbbm{1} \big|_{ \mathcal{F}_n^\pm} . 
	$$
	 The Hamiltonian (\ref{eq:strucH}) and each of its $ m $-body parts are particle-number conserving, and hence a direct sum of  $ n $-particle Hamiltonians, 
	$$ H^{(m)} =  \bigoplus_{n\geq 0}  H_n^{(m)} . $$ 
	From the canonical commutation rules~\eqref{eq:canonical}, one derives the combinatorial identity
	\begin{align}
		N  a^\dagger_{x_1}\cdots a^\dagger_{x_m} a_{y_m}\cdots a_{y_1} & = m a^\dagger_{x_1}\cdots   a^\dagger_{x_m} a_{y_m} \cdots  a_{y_1} \nonumber\\
		& \hspace{-10pt}+ \sum_{\xi\in\Lambda} a_\xi^\dagger \! \left( a^\dagger_{x_1}\cdots a^\dagger_{x_m} a_{y_m}\cdots a_{y_1} \right)\! a_\xi \label{number_commute}
	\end{align}
	for any $ m \in \mathbb{N} $ and any $ x_1 , \dots , x_m , y_m , \dots , y_1 \in \Lambda $. For each $ m $-body Hamiltonian, this immediately implies
	\begin{equation}\label{IRid}
		H_{n+1}^{(m)} =   \frac{1}{n+1-m}  \sum_{x\in \Lambda } a_x^\dagger\;  H_{n}^{(m)}\;  a_x 
	\end{equation}
	for all  $ n \geq m $. 
	Dropping non-negative terms, we thus arrive at the following operator inequality, which will be the key to subsequent results. 
	\begin{key}
	If $ H^{(m)} \geq 0 $ for all $ m > m_0 $, then  for all  $ n \geq m_1 $:
	\begin{equation}\label{IR}
		H_{n+1} \geq    \frac{1}{n+1-m_0}  \sum_{x\in \Lambda} a_x^\dagger\;  H_{n}\;  a_x .
	\end{equation}
	\end{key}
	We note that by invoking appropriate commutation relations, the relation \eqref{number_commute} and the Key Lemma can be adapted to the spin-$1/2$ quantum spin chain setting, and the spectral gap technique developed here can also be generalized to this context.
	
	The main purpose of this note is to investigate some simple consequences of this operator inequality for the structure of the low-energy spectrum of the $n $-particle Hamiltonians.\footnote{Of course, \eqref{IR} can also be used to compare higher eigenvalues.}
	To this end, we denote by 
	$$
	H_{n} = \sum_{\alpha \geq 0}  E_{n}^{(\alpha)}  \ P_{n}^{(\alpha)} 
	$$
	the spectral resolution on the $ n $-particle Fock space $ \mathcal{F}_n^\pm $, that is, $  E_{n}^{(0)} <  E_{n}^{(1)} < \dots $ denote the eigenvalues, which may be degenerate, and 
	$ P_n := P_{n}^{(0)}  $, $ P_{n}^{(1)}, \dots  $  are the corresponding eigenprojections of $ H_n $.

	\subsection{Recursive spectral relations}
	The following result states recursive spectral relations (I) and (II) between successive particle number sectors.
	

	\begin{theorem}[Recursive spectral relations]\label{thm:spectralcomp}
	Let $H$ be of the form \eqref{eq:strucH} with $ H^{(m)} \geq 0 $ for all $ m > m_0 $ and $H^{(m_0)}$ self-adjoint and bounded from below.
	For all $ n \geq m_1 $, we have the following two statements.
	\begin{enumerate}
	\item[(I)] The ground-state energy satisfies
	\begin{equation}\label{eq:gsind}
	 E_{n+1}^{(0)} \geq \frac{n+1}{n+1 - m_0}  E_{n}^{(0)} + \frac{ E_{n}^{(1)}-  E_{n}^{(0)}}{n+1 - m_0} \left( n+1 - \left\| G_n \right\| \right) 
	\end{equation}
	where 
	$$
	 G_n := P_{n+1}  \sum_{x\in \Lambda} a_x^\dagger\;  P_{n}\, a_x\;  P_{n+1} 
	 $$
	 is defined on $  \mathcal{F}_{n+1}^\pm $.
	 \item[(II)] The first excited energy satisfies
	\begin{equation}\label{eq:exind}
	 E_{n+1}^{(1)} \geq\frac{ n+1 -  \left\| F_n \right\| }{n+1 - m_0} \;  E_{n}^{(1)}  
	\end{equation}
	where 
\begin{equation}\label{eq:Fndefn}
	 F_n := P_{n+1}^\perp \sum_{x\in \Lambda} a_x^\dagger\;  P_{n}\;  a_x  P_{n+1}^\perp , \qquad P_{n+1}^\perp := 1 - P_{n+1} 
\end{equation}
	 is defined on $  \mathcal{F}_{n+1}^\pm $.
	\end{enumerate}
	\end{theorem} 
	\begin{proof}
	The assertions follow by combing~\eqref{IR} with the spectral estimate $ H_n \geq E_{n}^{(0)}  \mathbbm{1} + \left( E_{n}^{(1)}-  E_{n}^{(0)}\right) \left( \mathbbm{1}- P_n \right) $. 
	This yields
	\begin{align*}
		& H_{n+1} \geq \\  & \frac{1}{n+1-m_0} \left((n+1) E_{n}^{(0)}  \mathbbm{1} +  \left( E_{n}^{(1)}-  E_{n}^{(0)}\right) \left[ (n+1)  \mathbbm{1} -  \sum_{x\in \Lambda} a_x^\dagger\;  P_{n}\;  a_x  \right] \right) . 
	\end{align*}
	Projection onto the ground-state eigenspace $ P_{n+1} $, respectively, its orthogonal complement $  P_{n+1}^\perp  $, thus yields~\eqref{eq:gsind}, respectively,~\eqref{eq:exind}. 
	\end{proof}

	As suggested in the introduction, the recursive spectral relation~\eqref{eq:exind} can serve as a novel inductive method for deriving spectral gaps, a kind of martingale method \cite{nachtergaele1996spectral} with respect to particle number.   
It is loosely inspired by the proof of a spectral gap for Kac's master equation in the setting of classical kinetic theory \cite{Carlen:2003aa}.  %

	If we additionally assume that $ H^{(m_0)} \geq 0$ and that $ E_n^{(0)}  = 0 $ holds at least up to a maximal particle number $ n \leq n_{\max}$, then the bound~\eqref{eq:exind} becomes an inductive relation for the spectral gap 
	\begin{equation}\label{eq:neutralgap}
\gap_n := E_{n}^{(1)} - E_n^{(0)} 
	\end{equation}  
	Hence, we have identified a novel analytical tool to derive the positivity of the spectral gap up to $n_{\max}$.  Examples of systems where this version may be useful are FQHS and we elaborate on this  further at the end of Section~\ref{sec:FQHS}.


	\subsection{Relating the charge and neutral gap} \label{sec:cngap}
In addition to producing a lower bound on $\gap_n$ for nonnegative Hamiltonians with zero energy ground states, Theorem~\ref{thm:spectralcomp} has another physical interpretation, which we now highlight. 

Fix a site-set $ \Lambda $ and a particle number $ n_0 $ finite, but arbitrary.  
For the Hamiltonian $H_{n_0} $, one may compare the energy cost of 
creating an excitation at the same particle number, i.e. $\gap_{n_0}$,
with the energy cost of creating a hole, $$  \Delta^-_{n_0} :=  E_{n_0-1}^{(0)} - E_{n_0}^{(0)} $$
and a particle, $$  \Delta^+_{n_0} :=   E_{n_0+1}^{(0)} - E_{n_0}^{(0)}. $$  The total cost of creating a  particle-hole pair is hence
\begin{equation} \label{eq:chargegap}
	\Delta_{n_0} :=   \Delta^+_{n_0} +   \Delta^-_{n_0} =  E_{n_0+1}^{(0)} + E_{n_0-1}^{(0)}  - 2  E_{n_0}^{(0)} .
\end{equation}
Here, $\Delta_{n_0}$ is called the charge gap and $\gap_{n_0}$ is referred to as the neutral gap at particle number $ n_0 $. In the case of FQHS, which are repulsive systems, one has $E_{n}^{(0)}=0$ and $\Delta_{n}^-=0$ for all $n$ smaller than some $ n_{\max} $. The following corollary to Theorem~\ref{thm:spectralcomp} thus produces a relation between the charge and neutral gap for FQHS, which will be discussed further in Section~\ref{sec:FQHS}.
\begin{corollary}\label{cor:cngap}
	For Hamiltonians of the form ~\eqref{eq:strucH} with $ m_0 \geq 1 $ and non-negative interactions, that is, $ H^{(m)} \geq 0 $ for all $ m \geq 2 $, at any particle number $ n_0 \geq m_1 $ one has
	\begin{equation}\label{eq:chargeneutral}
		\Delta_{n_0}^+ \geq \frac{ E_{n_0}^{(0)}}{n_0} +  \frac{n_0 + 1 - \| G_{n_0} \|}{n_0} \  \gap_{n_0}  . 
	\end{equation} 
\end{corollary}
\begin{proof}
	This is immediate from using $m_0 \geq 1$ in \eqref{eq:gsind} to bound $\Delta_{n_0}^+$.
\end{proof}

Some remarks are in order:
	\begin{enumerate}[label=(\roman*)]
	\item
In the case $ E_{n_0}^{(0)}  \geq 0 $, e.g.\ for non-negative Hamiltonians, if $  \left\| G_{n_0} \right\|  $ is bounded uniformly in $ n_0 $, Corollary~\ref{cor:cngap} proves that for repulsive systems it is more costly to introduce an excess particle than to create a neutral excitation. If, in addition, the spectral gap vanishes for large $ \Lambda $, the bound~\eqref{eq:chargeneutral} reduces to the statement that $ \Delta_{n_0}^+ \geq 0$. 
\item If $ E_{n_0}^{(0)} $ stands for the minimal ground-state energy among all particle sectors, then the convexity of $ n \mapsto E_{n}^{(0)} $ and the normalization $ E_0^{(0)} = 0 $ would imply that $   \Delta_{n_0}^- \leq - E_{n_0}^{(0)}/ n_0 $. 
\item One can impose  suitable additional conditions under which the thermodynamic limit exists and the relation between the two types of spectral gaps persists with respect to the GNS Hamiltonian. This is rather straightforward for lattice fermions. In the case of bosons, additional technical issues arise which would take us beyond the scope of this note. 
\end{enumerate}

\subsection{General bounds on $\|G_n\|$}
	
	Any application of Theorem~\ref{thm:spectralcomp} or Corollary~\ref{cor:cngap} requires  upper bounds on the operator norms $ \left\| G_n \right\| $, respectively $ \left\| F_n \right\| $. While it is generally hard to control the overlap of ground states, one can easily derives robust estimates on the norm of $G_n$ in terms of a many-body Gram matrix, which at least in the fermionic case are useful.
	
	In this context, it is helpful to note that if $  \left( \varphi_\alpha \right)_{1\leq \alpha\leq q_n} $ stands for an orthonormal basis of the ground-state eigenspace, $ \ran P_n $, then $ \| G_n \| $ is bounded from above by the operator norm of 
	$
	 \sum_{x \in \Lambda } a_x^\dagger P_n a_x $, which in turn agrees with the norm of the $q_n |\Lambda| \times q_n |\Lambda|$ many-body Gram matrix $ G^{(n)} $with entries given by the scalar product
\begin{equation}\label{lem:gramnormferm}
G^{(n)}_{\alpha x ,\beta y} := \langle a_x^\dagger  \varphi_\alpha  , \ a_y^\dagger  \varphi_\beta \rangle . 
\end{equation}
The following is a straightforward bound. 
\begin{lemma}\label{lem:Gbound}
Let $ q_n = \dim \ran P_n $. Then, for all $ n $,
	\begin{align*}
	 \left\| G_n \right\| \leq   \big\| G^{(n)} \big\| \leq \begin{cases} q_n  & \hfill \mbox{(fermions)} \\
		n q_n  & \hfill  \mbox{(bosons)}
	\end{cases}
	\end{align*}
\end{lemma}
\begin{proof}
In the fermionic case, we start from the estimate 
\begin{equation}\label{eq:firstest}
  \left\| G_n \right\| \leq  \big\| G^{(n)} \big\| =  \left\|  \sum_{x \in \Lambda } a_x^\dagger P_n a_x \right\| \leq \sum_{\alpha=1}^{q_n} \left\|  \sum_{x \in \Lambda } a_x^\dagger |\varphi_\alpha\rangle\langle  \varphi_\alpha | a_x \right\| ,
\end{equation}
where we use the notation $  |\varphi_\alpha\rangle\langle  \varphi_\alpha|  $ for the rank-one projection onto one of the orthonormal, zero-energy eigenstates of $ H_n $. For fixed $ \alpha $, the operator norm on the right-hand side
agrees with the operator norm of the $|\Lambda| \times  |\Lambda|$-many-body Gram matrix  $ G_\alpha^{(n)} $ with entries indexed by $ x , y \in \Lambda $:
$$
 \big(G_\alpha^{(n)}\big)_{x,y} :=  \langle a_x^\dagger  \varphi_\alpha  | \ a_y^\dagger  \varphi_\alpha \rangle  = \delta_{x,y} -  \langle a_y  \varphi_\alpha  | \ a_x \varphi_\alpha \rangle .
$$
where we used  the CAR~\eqref{eq:canonical}. Since the last term also defines a Gram matrix, which is positive definite, we have $ G_\alpha^{(n)} \leq \mathbbm{1} $ and hence $ \| G_\alpha^{(n)} \| \leq 1 $. This completes the proof in case of fermions.

In the bosonic case, we us the rough estimate
$$  \left\| G_n \right\| \leq   \big\| G^{(n)} \big\| \leq   \tr G^{(n)}  = \sum_{\alpha=1}^{q_n} \langle \varphi_\alpha | N \varphi_\alpha \rangle = n q_n ,
$$
which completes the proof.
\end{proof} 
We will derive refinements of these general bounds on $\|G_n\|$, in particular in the bosonic case, based on symmetry considerations for FQHS in Section~\ref{sec:FQHS}

	\section{Fractional Hall systems}\label{sec:FQHS}
	
	The crude bounds on $\|G_n\|$ from Lemma~\ref{lem:Gbound} can be refined for FQHS by taking advantage of the common symmetries and properties that all such systems exhibit: translation invarience, charge and dipole conservation, and a maximal fractional filling of the ground state. We specifically consider FQHS in the torus geometry, for which the Landau levels generating the one-particle Hilbert space can be naturally identified with the sites of a one-dimensional ring $ \Lambda = [ 1:L ] :=\mathbb{Z}/L\mathbb{Z}$, see, e.g., \cite{Ortiz:2013aa} for more details. 
	
	In the case of a bosonic or fermionic system as in \eqref{eq:canonical} defined on a finite ring $\Lambda=[1:L]$, the many-body translation operator is given by the unitary $T$ such that
	\begin{equation}\label{eq:T}
	T^\dagger a_x T = a_{(x-1) \bmod L } \quad \mbox{for all $ x \in [1:L] $}.
	\end{equation} 
	The charge and dipole symmetries are encoded by the unitaries
	\[ 
	U = \exp\left( \frac{2\pi i}{L} N \right)  , \quad V  = \exp\left( \frac{2\pi i}{L} D \right)  
	\]
	where $U$ is generated by the particle-number operator $ N  = \sum_{x=1}^L  N_x $ with $ N_x= a_x^\dagger a_x  $, and $V$ is generated by the  
	dipole operator $ D = \sum_{x=1}^L x \ N_x  $.
	The interplay of these three symmetries is captured by the relations
	\begin{equation}\label{eq:rel}
		VT = UTV , \qquad UT = TU , \qquad UV = VU .
	\end{equation}
	This algebra and its generalizations have been studied in~\cite{Gromov:2019aa}.
	
In the case that $H_n \geq 0$ for all $n$ and $\ker H$ is nontrivial, the maximal fractional filling of the ground state can be expressed by the existence of coprime $1\leq p<q\in \mathbb{N}$ and a particle number $n_q\in\mathbb{N}$ such that \[p/q=n_q/L\] for which the ground state space of $H_{n_q}$ is given by $\ker H_{n_q}$ while $H_n >0$ for all $n >n_q$. As we will show below, this property is actually a consequence of another observed property of FQHS: namely, that the ground state space $\ker H_{n_q}$ is $q$-fold degenerate and generated by a fixed state $\varphi$ and the translation operator $T$. This motivates the following assumption we impose for our subsequent results. While we are specifically motivated by FQHS, our results are valid for any Hamiltonian that satisfies the following assumption.
	
	\begin{assum}\label{A2}
		The Hamiltonians are assumed to be of the form~\eqref{eq:strucH}  such that:
		\begin{itemize}
		\item \textbf{Symmetries.} $ H $ commutes with  $ T , U$, and $V$. 
		\item \textbf{Positivity.} There are coprime $1\leq p<q\in\mathbb{N}$ such that $n_q=pL/q\in \mathbb{Z}$ and $H_{n_q}\geq 0$. Moreover, $H_{n_q}^{(m)}\geq 0$ for all $m>m_0$ if $m_1>m_0$.
		\item \textbf{Fractionally filled ground state.} The ground state at $ n_q $ is of the form
		  $$ \ker H_{n_q} = \spa\left\{ \varphi , T  \varphi ,  \dots , T^{q-1} \varphi \right\} $$
		  for some normalized $ \varphi \in\ker H_{n_q} $ such that
		\begin{itemize}
			\item $ \varphi $ is an eigenvector of $ U $ and $ V $.
			\item $ \varphi $ is $ q $-periodic: \quad $ T^q \varphi = \varphi $. 
		\end{itemize}
	\end{itemize}
	\end{assum}	
	These assumptions encompass all $ m $-body pseudopotential Hamiltonians in fractional Hall physics. In the $ 2 $-body case $ m_0=m_1 = 2 $, these take the form~\cite{Trugman:1985lv,Pokrovsky_1985,Duncan1990,Lee:2015aa}:
	\begin{align}\label{eq:HPseudo}
		H  & =  \sum_{s \in \frac{1}{2} [1,2L] } Q_s^\dagger Q_s\\
		Q_s &  = \sum_{\substack{0\leq k\leq L: \\s-k\in\mathbb{Z}}} F(k)\ a_{(s-k) \bmod L}\  a_{(s+k) \bmod L}  . \notag 
	\end{align}
	where additions are understood modulo $ L $. 
	Hamiltonians of the form~\eqref{eq:HPseudo} are parent Hamiltonians for all basic filling fractions of the form $ 1/q$ (with $ q \geq 3$ odd for fermions, and $ q \geq 2 $ even for bosons~\cite{Regnault:2003aa}). In this case, $ \varphi $ plays the role of the Laughlin wavefunction~\cite{PhysRevLett.50.1395,Haldane:1985aa}, and $ T $ is the many-body magnetic translation on the torus~\cite{Haldane:1985ab,Frem15}. Choosing a torus geometry ensures the required ground state degeneracy \cite{Haldane:1985ab,Seidel:2005ab,Frem15,Burnell:2024aa}. We recall that FQHS physics (including gap properties in the thermodynamic limit) is expected to be robust with respect to the geometry \cite{haldane2018origin}. The class of Hamiltonians covered by the above assumptions also encompasses all other $ 2$-body pseudopotential Hamiltonians, including those whose ground states go beyond the basic filling fractions, as demonstrated, e.g., by the Jain-states~\cite{Chen:2017aa,Bandyopadhyay:2020aa}. The general $m$-body form of \eqref{eq:strucH} is relevant in the description of more exotic fractional quantum Hall states such as the Gaffnians~\cite{Lee:2015aa}. Finally, truncated pseudopotentials \cite{Bergholtz:2005pl,Nakamura:2012bu,PhysRevB.85.155116,NWY2020b,nachtergaele:2021a,warzel:2021,warzel:2022} are also within our framework, as well as their  limits~\cite{kapustin:2020}. In all these examples, the Hamiltonian is in fact positive with $ H_n^{(m)} \geq 0 $ for all $ n ,m $. 
	
	\subsection{Incompressibility} 
	In the above setup, the translates  $ T^j \varphi  $,  $ j=0,\ldots,q-1$, form an orthonormal basis of the ground state at $ n_q $,
		\begin{equation}\label{eq:gsp}
		\dim  \ker H_{n_q} = q ,
	\end{equation}
	reflecting the known, exact  $q$-fold degeneracy of the FQH-ground state at maximal filling on the torus~\cite{Haldane:1985ab,Seidel:2005ab,Frem15,Burnell:2024aa}. 
	This is easily seen by observing that, since $p$ and $q$ are coprime, these states are $ V $-eigenstates with distinct eigenvalues:
	\begin{equation}\label{eq:eigenvaluesV}
		V T^j \varphi  =   \exp\left( 2\pi i  \ \Big( \frac{d_\varphi }{L} + \frac{jp}{q} \Big) \right)  T^j \varphi ,
	\end{equation} 
where $ d_\varphi \in \mathbb{Z} $ the $ D $-eigenvalue corresponding to $ \varphi $.\\ 
	
	\begin{lemma}\label{lem:zeroall}
Suppose that $H$ satisfies Assumption\ref{A2} for some $n_q>m_1$. If additionally $H_{n}^{(m)} \geq 0$ for all $ m $ and $ n < n_q $, then for all $ m_1 \leq n \leq n_q $:
	$$   E^{(0)}_{n} = 0  $$
	\end{lemma}
	\begin{proof} 
	This follows from the fact that $ E_{n_q}^{(0)} = 0 $ by successivley applying the inequality \eqref{IR} starting with $ n = n_q-1 $. 
	\end{proof}
	It is a hallmark of a system's incompressibility that upon adding an extra particle, the ground-state energy rises. In particular, in the situation of Lemma~\ref{lem:zeroall}, one has 
	$$ \Delta_{n_q} = E_{n_q+1}^{(0)} . $$
	As the following theorem illustrates, all systems satisfying the above assumptions are incompressible at $n_q$ since for any $ n > n_q $, the ground-state energy of $H_n$ is above zero, i.e., the ground-state energy of $H_{n_q}$:
	
	\begin{theorem}[Incompressibility at $ n_q $] \label{thm:nokern}
		For any $ n \geq n_q $ there are no zero energy eigenstates of $ H_{n+1} \geq 0$, i.e. 
		\begin{equation}\label{eq:nogs}
			\dim  \ker H_{n+1} = 0 , 
		\end{equation}
		as long as the system is sufficiently large: $ L \geq q^2/p $ in the case of fermions, and $ L \geq (1+p) q^2/p $ in the case of bosons. 
	\end{theorem}
	\begin{proof}
	Given  \eqref{IR} for any $ n \geq n_q $, it suffices to consider the case $ n = n_q $.
	By way of contradiction, suppose that $ \psi \in \ker H_{n_{q}+1} $ is a normalized eigenfunction of $ V $ with eigenvalue $ \exp\left( 2\pi i  \ \frac{d_\psi}{L}  \right) $ for some $ d_\psi \in \mathbb{Z} $. Then, by~\eqref{IR}, we also have $ a_x \psi  \in \ker H_{n_q} $ for all $ x $, and 
	\begin{equation}\label{eq:Vshifta}
		V a_x\psi = \exp\left( 2\pi i  \ \frac{d_\psi - x }{L}  \right)  a_{x} \psi .
	\end{equation}
	This implies that any non-trivial $ a_x \psi \neq 0 $ agrees --- up to a phase and the normalization factor $ \| a_x \psi \| $ --- with one of the vectors in $\left\{ \varphi , T  \varphi ,  \dots , T^{q-1} \varphi \right\} $, and the eigenvalue
	$
	\exp\left( 2\pi i  \ \frac{d_\psi - x }{L}  \right) 
	$
	must coincide with one of the eigenvalues \eqref{eq:eigenvaluesV} of the eigenvectors of $V$ in $\ker H_{n_q}$. There are only $q$ values of $x$ for which this is possible. Hence, $a_x\psi\neq 0$  for at most $q$ values of $x$. Moreover,
	\begin{equation}\label{particlenumberbound}
		n_q+1 = \langle\psi , N \psi\rangle =\sum_{x=1}^L  \langle a_x\psi , a_x\psi\rangle \leq q \max_{x}  \langle\psi , N_x \psi\rangle.
	\end{equation}
	For fermions this implies $n_q+1 \leq q$, a contradiction since $L\geq \frac{q^2}{p}$. 
	
	For bosons, we estimate $ \langle\psi , N_x \psi\rangle$ as follows. First, the $ q $-periodicity of $ \varphi $ implies 
	\begin{equation}\label{eq:boundNj}
		\max_{1\leq y\leq L}  \langle \varphi , N_y  \varphi \rangle \leq  \sum_{y=1}^{q} \langle \varphi , N_y\varphi \rangle = \frac{q}{L}   \langle \varphi , N  \varphi \rangle = p.
	\end{equation}
	Since $a_x\psi \in \ker H_{n_q}$, we thus conclude
	\begin{align*}
		p\geq & \  \frac{\langle a_x\psi, N_x a_x\psi\rangle}{\langle a_x\psi, a_x\psi\rangle}
		=  \frac{\langle \psi, N_x^2 \psi\rangle-\langle \psi, N_x\psi\rangle}{\langle \psi, N_x\psi\rangle} \\
		\geq & \ \frac{\langle \psi, N_x \psi\rangle^2-\langle \psi, N_x\psi\rangle}{\langle \psi, N_x\psi\rangle}
		= \langle \psi, N_x\psi\rangle -1 .
	\end{align*}
	This implies
	$
	\langle \psi, N_x\psi\rangle\leq p+1  $, and so \eqref{particlenumberbound} yields $n_q+1 \leq (p+1)q$, a contradiction to the assumption $n_q \geq(p+1) q $. This concludes the proof of Theorem~\ref{thm:nokern}.
	\end{proof}

		\subsection{Charge and neutral gap in FQHS}  
	Comparing~\eqref{eq:gsp} and \eqref{eq:nogs}, the ground-state energy at $ n_q $ vs.\ $ n_q+1 $ exhibits a jump, which agrees with $
	 \Delta_{n_q}^+ = E^{(0)}_{n_q+1} > 0 $.  In comparison, 
	the spectral gap, $ \gap_{n_q} = E^{(1)}_{n_q} - E^{(0)}_{n_q} = E^{(1)}_{n_q} $, is the neutral gap at maximall filling and conjectured to be the lowest excited energy in the system.

	As a main result, we establish that in FQHS the charge gap is at least as large as the neutral gap. 
	\begin{theorem}[Charge vs. neutral gap  in FQHS] \label{thm:main}
		In the case of $q$-commensurate maximal filling, i.e. $ L = \ell q^2 $ for some $ \ell\geq 3  $, the charge gap dominates the neutral gap, for any $ n \geq n_q $
		\begin{align*}
			& E^{(0)}_{n+1}\geq \frac{n_q}{n_q+1-m_0}  \ \gap_{n_q}  & \mbox{(fermions)} \\
			&E^{(0)}_{n+1} \geq \frac{n_q-p }{n_q+1-m_0}  \  \gap_{n_q}     & \mbox{(bosons)} 
		\end{align*} 
	\end{theorem} 
	\begin{proof} Both assertions are a consequence of the spectral comparison Theorem~\ref{thm:spectralcomp}. 
	The bound~\eqref{eq:gsind} already implies $ E^{(0)}_{n+1} \geq \frac{n+1}{n+1-m_0}  E^{(0)}_{n} $ for all $n\geq n_q+1$. Therefore, it suffices to prove the claim for $n=n_q$, in which case the inequalities follow from~\eqref{eq:gsind}  and Lemma~\ref{lem:GFQHE} below. 
	\end{proof}

The above theorem is an 
	 improvement over Corollary~\ref{cor:cngap}, in particular in the bosonic situation. As we will see in Lemma~\ref{lem:GFQHE}  below, it rests on utilizing the three symmetries found in FQHS. Before turning to this proof, 
a few remarks are in order.
	\begin{enumerate}[label=(\roman*)]
	\item  For bosons, the bound from Theorem~\ref{thm:main} implies that the charge gap dominates the neutral gap if $m_0\geq p+1 $.   For fermions, the bound is conclusive even for $ m_0 = 1 $, which allows for the presence of a suitably adjusted, negative chemical potential $ \mu N  $ in the Hamiltonian. Moreover, as the proof of Lemma~\ref{lem:GFQHE} shows, the commensurability condition $ L = \ell q^2 $ can be dropped for $p=1$ and $q=2$.

\item	There is evidence that in geometries without translation invariance finite-size effects may cause a violation of the inequalities \cite{balram2024fractional},   but it is expected that the gap domination result always holds in the bulk. More precisely, by considering a sequence of system sizes of the form $L=\ell q^2$ tending to infinity, one sees that our the domination of the charge gap over the neutral gap extends to the thermodynamic limit. 
	
	\item 
	The inequality among the two gaps is shared by the nematic fractional Hall phases~\cite{Fradkin:1999aa,Regnault:2017aa, PhysRevResearch.2.033362,Pu:2024aa}. The characteristic of the nematic phase is that the neutral gap vanishes while the charge gap remains open. 
	Since the above theorem covers such systems, one cannot expect (nor do we assume for the validity of our result) that the neutral gap is generally uniformly positive.

	\item For Haldane pseudopotentials~\eqref{eq:HPseudo}, our lower bounds may be combined and compared to the upper bounds on the charge gap derived in~\cite{Weerasinghe:2016aa}:
	\[
	\gap_{n_q+1}  \leq \begin{cases} \frac{4p}{q-p}\ \Delta & \hfill \mbox{(fermions)} \\
		\frac{4p}{q+p} \ \Delta & \hfill  \mbox{(bosons)}
	\end{cases}
	\]
	where $ \Delta =  \frac{1}{4}   \sum_{s} \sum_{k:s-k\in\mathbb{Z}} \left| F(k) \right|^2  $.\\ %
		\end{enumerate}
	The key to the proof of  Theorem~\ref{thm:main} is the following considerable improvement of Lemma~\ref{lem:Gbound} in the FQHS setup. 
	\begin{lemma}\label{lem:GFQHE}
	In the situation of Theorem~\ref{thm:main}:
	\begin{align*}
	 \left\| G_{n_q} \right\| \leq   \big\| G^{(n_q)} \big\| \leq \begin{cases} 1  & \hfill \mbox{(fermions)} \\
		1+ p   & \hfill  \mbox{(bosons)}
	\end{cases}
	\end{align*}
	\end{lemma}
	\begin{proof}
	The first inequality agrees with the first step in~\eqref{eq:firstest}. It is thus left to estimate the operator norm of the many-body Gram matrix $ G := G^{(n_q)} $ with entries 
	\[
	G_{xk,x'k'} =   \langle  a_x^\dagger T^{k-1} \varphi,  a_{x'}^\dagger T^{k'-1} \varphi\rangle ,
	\]
	indexed by $ x ,x'\in[1:L]$, $ k,k' \in [1:q] $, which is accomplished using the symmetries of the system. Denoting the $D$-eigenvalue of $ \varphi $ by $ d_\varphi  $, we have 
	\[ V  a_x^\dagger T^{k-1} \varphi =  \exp\left( 2\pi i \frac{d_\varphi + (k-1) n_q x }{L}\right) a_x^\dagger T^{k-1} \varphi  .
	\]
	This implies that 
	$ \langle  a_x^\dagger T^{k-1} \varphi,  a_{x'}^\dagger T^{k'-1} \varphi\rangle = 0  $ unless $ \left( k n_q +x\right) \bmod L = \left( k' n_q +x' \right) \bmod L  $. 
	The Gram matrix is hence block diagonal  
	\[ 
	G =  \bigoplus_{\gamma = 1}^L \ G(\gamma)  
	\] 
	with $ q\times q $ blocks  
	$$
		G(\gamma)_{k,k'}   = \left\langle a_{(\gamma-n_q k )\bmod L}^\dagger \ T^{k -1} \varphi ,  a_{(\gamma-n_q k' )\bmod L}^\dagger \ T^{ k'-1} \varphi \right\rangle . 
	$$
	In turn,  the (anti-)commutation rules imply that each block has the following structure $
	G(\gamma)   =  \mathbbm{1} \pm F(\gamma) $ 
	with $ (-) $ for fermions and  $ (+) $ for bosons and, by translating,
	\begin{align*}
		F(\gamma)_{k,k'} =    \langle &a_{(\gamma+n_qk )\! \!\bmod\! L} \ T^{k -1 +(k+k') n_q} \varphi ,\\  &a_{(\gamma+n_q k') \!\!\bmod \! L} \ T^{ k'-1+ (k+k') n_q} \varphi \rangle.
	\end{align*}
	Note that $n_q=p L/q=p\ell q$ is a multiple of $q$. Hence, the $ q $-periodicity of $ \varphi $ implies that 
	$ F(\gamma) $ 
	is another Gram matrix, such that $ F(\gamma) \geq 0 $. For fermions, this already suffices because it implies $ 0 \leq    G(\gamma)   \leq    \mathbbm{1}  $ and so $\|G\|\leq 1$.
	
	For bosons, \eqref{eq:boundNj} gives
	\begin{align*} \tr F(\gamma) = & \sum_{x=1}^q\left\langle \varphi , N_{(\gamma+1 -(n_q+1) x) \bmod p} \ \varphi \right\rangle \\
		= &  \sum_{x=1}^q\left\langle \varphi , N_{x} \ \varphi \right\rangle = p.
	\end{align*}
	This implies $ 0 \leq F(\gamma) \leq p \mathbbm{1} $ and hence $ 0 \leq    G(\gamma)   \leq  (1+p)  \ \mathbbm{1}  $, which proves the claim.
	\end{proof}
	
		We conclude by briefly describing  the novel approach  to the Haldane FQHS gap conjecture that follows by iterating Relation (II).
We  recall that for pseudopotentials~\eqref{eq:HPseudo} corresponding to $ 1/q$-filling with small $ q \in \{ 2, 3 , \dots \} $, Haldane's FQHS conjecture states that the excitation gap above the ground-state energy is uniformly positive, $ \gap_{nq}\geq c > 0 $. Note that this is not in contradiction to Lieb-Schultz-Mattis type results for systems with dipolar symmetry~\cite{Oshikawa:2000aa,Burnell:2024aa}.  For larger values of $ q $, the neutral gap is conjectured to close~\cite{PhysRevLett.54.237}. 	

As an approach to this problem, one may iterate \eqref{eq:exind} staring from $ n = m_1 $ up to $n\leq n_q$. To obtain a positive gap, one needs a uniform upper bound on $$ \sum_{k\geq m_1}^{n_{\max}} \frac{\| F^{(k)} \| - m_0 }{k+1} .$$
Obtaining this will require a fine understanding of the matrix $F^{(k)}$ from \eqref{eq:Fndefn}, which is related to the overlap between a ground state with a particle added and the set of excited states.

	\section*{Acknowledgments}
	This work was supported by the  DFG under the grants TRR 352 – Project-ID 470903074 (ML, SW) and EXC-2111-390814868 (SW), and by the National Science Foundation under grant DMS-2108390 (BN). All authors acknowledge support through a SQuaRE collaboration grant from the American Institute of Mathematics at Caltech, where part of this work was accomplished. BN and AY thank the TU Munich for kind hospitality during a visit in the course of this work. BN and SW thank the Max Planck Institute for Mathematics (Bonn) for hospitality during the final stages of this work. BN and SW thank Duncan Haldane for valuable comments. We thank Ting-Chun Lin and the anonymous referees for useful comments. 
	\bibliographystyle{abbrv}  
	\bibliography{CNGv9} 
	\end{document}